\documentclass{article}%
\usepackage{amsmath}
\usepackage{amsfonts}
\usepackage{amssymb}
\usepackage{graphicx}%
\providecommand{\U}[1]{\protect\rule{.1in}{.1in}}
\newtheorem{theorem}{Theorem}

\newtheorem{lemma}[theorem]{Lemma}

\newenvironment{proof}[1][Proof]{\noindent\textbf{#1.} }{\ \rule{0.5em}{0.5em}}
\begin{document}

\title{On interaction-free measurement}
\author{Keiji Matsumoto\\National Institute of Informatics\\2-1-2 Hitotsubashi, Chiyoda-ku, Tokyo 101-8430\\e-mail : keiji@nii.ac.jp}
\maketitle

\section*{Abstract}

This manuscript is inspired by the paper \cite{Kwiat-etal}. In the paper, they
investigate a method to detect existence of an object with arbitrarily small
interaction. Below, we sketch their protocol to motivate the present
manuscript. The object of their protocol is to detect whether the given
blackbox interact with input states or not, with negligible distortion of the
blackbox, and high detection probability.

In this paper, we do two things. First, we prove the above mentioned protocol
is optimal in a certain setting. The main tool here is adversary
method\thinspace\cite{Ambainis}, a classical method in query complexity.
Second, we present a protocol to detect unitary operations with negligible
error and no distortion of the input at all. 

\section{Introduction}

This manuscript is inspired by the paper \cite{Kwiat-etal}. In the paper, they
investigate a method to detect existence of an object with arbitrarily small
interaction. Below, we sketch their protocol to motivate the present
manuscript. The object of their protocol is to detect whether the given
blackbox interact with input states or not, with negligible distortion of the
blackbox, and high detection probability.

They generate a photon, split it into two paths, one of which leads to the
given blackbox. After going through it, the beam is mixed again. After
repeating the process for many times, the photon is measured to decide whether
the blackbox is empty or not. The amplitude for the path going into the
blackbox stays arbitrarily small if the blackbox interacts with the photon,
thus satisfying almost interaction-free condition. Otherwise, this amplitude
grows almost to 1, enabling almost complete detection. If the probability of
correct detection is almost $1$, and the times of repetition is $T$, the
distortion of the system is $O\left(  1/T\right)  $.

In this paper, we do two things. First, we prove the above mentioned protocol
is optimal in a certain setting. The main tool here is adversary
method\thinspace\cite{Ambainis}, a classical method in query complexity.
Second, we present a protocol to detect unitary operations with negligible
error and no distortion of the input at all. The differences between the two
versions are : \ (1) whether the operation to be estimated is a unitary
operation acting only on the input particle $\mathcal{H}_{I}$ only, or with
certain degree of freedom $\mathcal{H}_{B}$ of the blackbox.

\ The first scenario: (1) the given blackbox is an interaction between
$\mathcal{H}_{I}$ and $\mathcal{H}_{B}$, and (2) the distortion is defined as
the change of the reduced density matrix on $\mathcal{H}_{B}$. With additional
assumption that the part of $\mathcal{H}_{B}$ interacting with $\mathcal{H}%
_{I}$ is refreshed for each time, we give the converse statement, insisting
that the protocol in \cite{Kwiat-etal} is optimal. In the second scenario, (1)
the given blackbox is a unitary operation acting only on $\mathcal{H}_{I}$,
and therefore (2) the distortion is defined as the change of reduced density
matrix on $\mathcal{H}_{I}$. This scenario may differ from the initial
motivation of "interaction-free" measurement, but obtained result is strong
enough to justify the setting: if the number of candidates for the given
unknown operation is finite, estimation is possible without any distortion at
all, and with arbitrarily small error.

\section{The first scenario}

Our intention here is to prove optimally of the protocol presented in
\cite{Kwiat-etal} in a reasonable setting. The main tool here is adversary
method\thinspace\cite{Ambainis}, a classical method in query complexity.

In this setting, the given black box is a controlled unitary interaction $U$
between $\mathcal{H}_{I}$ and $\mathcal{H}_{B}$, ($U\in\mathrm{U}\left(
\mathcal{H}_{I}\otimes\mathcal{H}_{B}\right)  $) and $U$ is either
$U_{0}=\mathbf{1}$ or $U_{1}$ ( $\Theta:=\left\{  0,1\right\}  $ ). \ We
suppose that the subsystem $\mathcal{H}_{B^{\prime}}$ of $\mathcal{H}_{B}$ is
interacting with $\mathcal{H}_{I}$. The state of the subsystem $\mathcal{H}%
_{B^{\prime}}$ is refreshed to $\left\vert f\right\rangle $ at every time
step. Put differently, $\mathcal{H}_{B}=\mathcal{H}_{B^{\prime}}^{\otimes n}$,
and the initial state is $\left\vert f\right\rangle ^{\otimes n}$, and each
time one of $\left\vert f\right\rangle $'s subject to $U$. Further,
$A_{\theta}:=\left\langle f\right\vert U_{\theta}\left\vert f\right\rangle $
satisfies
\begin{equation}
\left\Vert A_{1}\right\Vert <1.\label{|A|<1}%
\end{equation}
This means that $U_{1}$ correlates $\mathcal{H}_{C}$ and $\mathcal{H}_{B}$,
i.e.,
\[
U_{1}\neq U_{1,I}\otimes U_{1,B}.
\]
Control bit space is described by $\mathcal{H}_{C}$. Thus, we are given a
either $\mathbf{1}_{I}\otimes\mathbf{1}_{B}$ or $c-U_{1}$, both acting on
$\mathcal{H}_{C}\otimes\mathcal{H}_{I}\otimes\mathcal{H}_{B}$.

In addition to the given black-box operation, we do information processing
operations, state preparation, feedback and so on. The working space necessary
for these operations is denoted by $\mathcal{H}_{W}$. Thus, the whole system
is $\mathcal{H}_{W}\otimes\mathcal{H}_{C}\otimes\mathcal{H}_{I}\otimes
\mathcal{H}_{B}$. We suppose these information processing operation cannot
touch $\mathcal{H}_{B}$, i.e., it is a unitary operation on $\mathcal{H}%
_{W}\otimes\mathcal{H}_{C}\otimes\mathcal{H}_{I}$. The operations between
$t-1$-th and $t$-th application of $c-U$ is denoted by $V_{t}$.

Our concern is to keep $\mathcal{H}_{B}$ unchanged. So the measure of the
distortion is
\[
D_{t}:=\left\Vert \left.  \left\vert \Phi_{\theta,t}\right\rangle \left\langle
\Phi_{\theta,t}\right\vert \,\right\vert _{\mathcal{H}_{B}}-\left(  \left\vert
f\right\rangle \left\langle f\right\vert \right)  ^{\otimes t}\right\Vert
_{1},
\]
where $\left\vert \Phi_{\theta,t}\right\rangle $ is a state of the whole
system (including the control bits) at the time step $t$ . After the $T$-th
application of the $c-U_{\theta}$, we measure the state, but only its
$\mathcal{H}_{C}\otimes\mathcal{H}_{W}\otimes\mathcal{H}_{I}$-part, or
equivalently, the measurement cannot touch $\mathcal{H}_{B}$. Therefore,
\[
\left\Vert \mathrm{tr}\,_{\mathcal{H}_{B}}\left\vert \Phi_{0,T}\right\rangle
\left\langle \Phi_{0,T}\right\vert -\mathrm{tr}\,_{\mathcal{H}_{B}}\left\vert
\Phi_{1,T}\right\rangle \left\langle \Phi_{1,T}\right\vert \right\Vert
_{1}\geq1-\varepsilon
\]
has to hold for $\varepsilon$ with $0<\varepsilon<1$.

By the protocol in \cite{Kwiat-etal}, $D_{T}\,=O\left(  1/T\right)  $ is
achieved for large $T$. They do not use $\mathcal{H}_{W}$, and the measurement
does not touch $\mathcal{H}_{I}$. Even with these restriction, their protocol
turns out to be at least as good as any protocols with the additional work
space $\mathcal{H}_{W}$ and a measurement over $\mathcal{H}_{C}\otimes
\mathcal{H}_{W}\otimes\mathcal{H}_{I}$. \ More specifically, any protocol in
our framework should satisfy
\begin{equation}
D_{T}\geq\left(  C+1\right)  ^{-2}\left(  1-\varepsilon\right)  ^{2}\frac
{1}{T}=O\left(  \frac{1}{T}\right)  ,\label{D>1/T}%
\end{equation}
where
\[
C:=6\left(  1-\left\Vert A_{1}\right\Vert ^{2}\right)  ^{-1/2}.
\]
Here, observe by assumption(\ref{|A|<1}), $C$ is finite. 

Since $\left\vert \Phi_{\theta,t}\right\rangle $ is not tractable, we
introduce \
\begin{align*}
\left\vert \psi_{t}\right\rangle  &  :=\left\langle f^{\otimes t}\right.
\left\vert \Phi_{\theta,t}\right\rangle \\
&  =\left\langle f^{\otimes t}\right\vert \left(  c-U_{\theta}\right)
V_{t}\otimes\mathbf{1}_{B}\cdots\left(  c-U_{\theta}\right)  V_{2}%
\otimes\mathbf{1}_{B}\left(  c-U_{\theta}\right)  V_{1}\otimes\mathbf{1}%
_{B}\left\vert \varphi\right\rangle \left\vert f^{\otimes t}\right\rangle \\
&  =\left(  c-A_{\theta}\right)  V_{t}\cdots\left(  c-A_{\theta}\right)
V_{2}\left(  c-A_{\theta}\right)  V_{1}\left\vert \varphi\right\rangle ,
\end{align*}
with $V_{t}$ acting on $\mathcal{H}_{C}\otimes\mathcal{H}_{W}\otimes
\mathcal{H}_{I}$ , and $A_{\theta}:=\left\langle f\right\vert U_{\theta
}\left\vert f\right\rangle $. Also %

\begin{align*}
\left\vert \varphi_{t}\right\rangle  &  :=\left\langle f^{\otimes t}\right.
\left\vert \Phi_{0,t}\right\rangle =V_{t}\cdots V_{2}V_{1}\left\vert
\varphi\right\rangle ,\\
\left\vert \psi_{t}^{\prime}\right\rangle  &  :=V_{t}\left\vert \psi
_{t-1}\right\rangle ,
\end{align*}
so that $\left\vert \psi_{t}\right\rangle =c-A_{1}\left\vert \psi_{t}^{\prime
}\right\rangle $. 

The distortion $D_{t}$ is bounded from below as follows.
\begin{align*}
D_{t} &  =\left\Vert \left.  \left\vert \Phi_{1,t}\right\rangle \left\langle
\Phi_{1,t}\right\vert \,\right\vert _{\mathcal{H}_{B}}-\left(  \left\vert
f\right\rangle \left\langle f\right\vert \right)  ^{\otimes t}\right\Vert
_{1}\\
&  \geq1-\left\langle f^{\otimes t}\right\vert \left(  \left.  \left\vert
\Phi_{1,t}\right\rangle \left\langle \Phi_{1,t}\right\vert \,\right\vert
_{\mathcal{H}_{B}}\right)  \left\vert f^{\otimes t}\right\rangle \\
&  =1-\mathrm{tr}\left\langle f^{\otimes t}\right.  \left\vert \Phi
_{1,t}\right\rangle \left\langle \Phi_{1,t}\right\vert \left.  f^{\otimes
t}\right\rangle \\
&  =1-\left\Vert \left\langle f^{\otimes t}\right.  \left\vert \Phi
_{1,t}\right\rangle \right\Vert ^{2}\\
&  =1-\left\Vert \psi_{t}\right\Vert ^{2}.
\end{align*}
Also,%

\begin{align}
&  \left\Vert \mathrm{tr}\,_{\mathcal{H}_{B}}\left\vert \Phi_{0,t}%
\right\rangle \left\langle \Phi_{0,t}\right\vert -\mathrm{tr}\,_{\mathcal{H}%
_{B}}\left\vert \Phi_{1,t}\right\rangle \left\langle \Phi_{1,t}\right\vert
\right\Vert _{1}\nonumber\\
&  \leq\left\Vert \mathrm{tr}\,_{\mathcal{H}_{B}}\left\vert \Phi
_{0,t}\right\rangle \left\langle \Phi_{0,t}\right\vert -\left\langle
f^{\otimes t}\right.  \left\vert \Phi_{1,t}\right\rangle \left\langle
\Phi_{1,t}\right\vert \left.  f^{\otimes t}\right\rangle \right\Vert
_{1}+\left\Vert \left\langle f^{\otimes t}\right.  \left\vert \Phi
_{1,t}\right\rangle \left\langle \Phi_{1,t}\right\vert \left.  f^{\otimes
t}\right\rangle -\mathrm{tr}\,_{\mathcal{H}_{B}}\left\vert \Phi_{1,t}%
\right\rangle \left\langle \Phi_{1,t}\right\vert \right\Vert _{1}\nonumber\\
&  =\left\Vert \mathrm{tr}\,_{\mathcal{H}_{B}}\left\vert \Phi_{0,t}%
\right\rangle \left\langle \Phi_{0,t}\right\vert -\left\langle f^{\otimes
t}\right.  \left\vert \Phi_{1,t}\right\rangle \left\langle \Phi_{1,t}%
\right\vert \left.  f^{\otimes t}\right\rangle \right\Vert _{1}+1-\mathrm{tr}%
\left\langle f^{\otimes t}\right.  \left\vert \Phi_{1,t}\right\rangle
\left\langle \Phi_{1,t}\right\vert \left.  f^{\otimes t}\right\rangle
\nonumber\\
&  =\left\Vert \left\vert \varphi_{t}\right\rangle \left\langle \varphi
_{t}\right\vert -\left\vert \psi_{t}\right\rangle \left\langle \psi
_{t}\right\vert \right\Vert _{1}+1-\left\Vert \psi_{t}\right\Vert
^{2}.\label{dist<}%
\end{align}
Therefore, the key quantities can be bounded from safer side by tracking
$\left\vert \psi_{t}\right\rangle $ and $\left\vert \varphi_{t}\right\rangle
$. 

Below, for each $\left\vert \varphi\right\rangle \in\mathcal{H}_{C}%
\otimes\mathcal{H}_{W}\otimes\mathcal{H}_{I}\otimes\mathcal{H}_{B}$,
$\left\vert \varphi_{1}\right\rangle $ and $\left\vert \varphi_{2}%
\right\rangle $ is the part part which is subject to/not subject to
$U_{\theta}$, i.e., if $\left\{  \left\vert i\right\rangle ;i=1,2\right\}  $
is the CONS of the $\mathcal{H}_{C}$ corresponding to the control bit,
\[
\left\vert \varphi_{i}\right\rangle :=\left\langle i\right.  \left\vert
\varphi\right\rangle \text{ }\left(  i=1,2\right)  ,
\]
so that
\[
c-U_{\theta}\left[
\begin{array}
[c]{c}%
\varphi_{1}\\
\varphi_{2}%
\end{array}
\right]  =\left[
\begin{array}
[c]{c}%
\varphi_{1}\\
U_{\theta}\varphi_{2}%
\end{array}
\right]  .
\]
Then we have%

\begin{align*}
\left\Vert \psi_{t}\right\Vert ^{2}  &  =\left\Vert \psi_{t,1}^{\prime
}\right\Vert ^{2}+\left\Vert A_{1}\right\Vert ^{2}\left\Vert \psi
_{t,2}^{\prime}\right\Vert ^{2}\\
&  =\left\Vert \psi_{t}^{\prime}\right\Vert ^{2}-\left\Vert \psi_{t,2}%
^{\prime}\right\Vert ^{2}+\left\Vert A_{1}\right\Vert ^{2}\left\Vert
\psi_{t,2}^{\prime}\right\Vert ^{2}\\
&  =\left\Vert \psi_{t-1}\right\Vert ^{2}-\left\Vert \psi_{t,2}^{\prime
}\right\Vert ^{2}+\left\Vert A_{1}\right\Vert ^{2}\left\Vert \psi
_{t,2}^{\prime}\right\Vert ^{2},
\end{align*}
which implies%
\begin{equation}
\left(  1-\left\Vert A_{1}\right\Vert ^{2}\right)  \left\Vert \psi
_{t,2}^{\prime}\right\Vert ^{2}=\left\Vert \psi_{t-1}\right\Vert
^{2}-\left\Vert \psi_{t}\right\Vert ^{2}. \label{phi_2}%
\end{equation}

On the other hand, the first term of (\ref{dist<}) is bounded from above as
follows:%
\begin{align*}
&  \left\Vert \left\vert \varphi_{t}\right\rangle \left\langle \varphi
_{t}\right\vert -\left\vert \psi_{t}\right\rangle \left\langle \psi
_{t}\right\vert \right\Vert _{1}\\
&  \leq\left\Vert \left\vert \varphi_{t}\right\rangle \left\langle \varphi
_{t}\right\vert -\left\vert \psi_{t}^{\prime}\right\rangle \left\langle
\psi_{t}^{\prime}\right\vert \right\Vert _{1}+\left\Vert \left\vert \psi
_{t}^{\prime}\right\rangle \left\langle \psi_{t}^{\prime}\right\vert
-\left\vert \psi_{t}\right\rangle \left\langle \psi_{t}\right\vert \right\Vert
_{1}\\
&  =\left\Vert V_{t}\left(  \left\vert \varphi_{t-1}\right\rangle \left\langle
\varphi_{t-1}\right\vert -\left\vert \psi_{t-1}\right\rangle \left\langle
\psi_{t-1}\right\vert \right)  V_{t}^{\dagger}\right\Vert _{1}\\
&  +\left\Vert \left\vert \psi_{t}^{\prime}\right\rangle \left\langle \psi
_{t}^{\prime}\right\vert -\left(  c-A_{1}\right)  \left\vert \psi_{t}^{\prime
}\right\rangle \left\langle \psi_{t}^{\prime}\right\vert \left(
c-A_{1}\right)  ^{\dagger}\right\Vert _{1}\\
&  =\left\Vert \left\vert \varphi_{t-1}\right\rangle \left\langle
\varphi_{t-1}\right\vert -\left\vert \psi_{t-1}\right\rangle \left\langle
\psi_{t-1}\right\vert \right\Vert _{1}\\
&  +\left\Vert \left\vert \psi_{t}^{\prime}\right\rangle \left\langle \psi
_{t}^{\prime}\right\vert -\left(  c-A_{1}\right)  \left\vert \psi_{t}^{\prime
}\right\rangle \left\langle \psi_{t}^{\prime}\right\vert \left(
c-A_{1}\right)  ^{\dagger}\right\Vert _{1}.
\end{align*}
Observe, by \ $\left\Vert A_{1}\right\Vert <1$ and%

\begin{align*}
&  \left\Vert \left\vert \psi_{t}^{\prime}\right\rangle \left\langle \psi
_{t}^{\prime}\right\vert -\left(  c-A_{1}\right)  \left\vert \psi_{t}^{\prime
}\right\rangle \left\langle \psi_{t}^{\prime}\right\vert \left(
c-A_{1}\right)  ^{\dagger}\right\Vert _{1}\\
&  =\left\Vert \left[
\begin{array}
[c]{cc}%
0 & \left\vert \psi_{t,1}^{\prime}\right\rangle \left\langle \psi
_{t,2}^{\prime}\right\vert \left(  \mathbf{1}-A_{1}\right)  ^{\dagger}\\
\left(  \mathbf{1}-A_{1}\right)  \left\vert \psi_{t,2}^{\prime}\right\rangle
\left\langle \psi_{t,1}^{\prime}\right\vert  & A_{1}\left\vert \psi
_{t,2}^{\prime}\right\rangle \left\langle \psi_{t,2}^{\prime}\right\vert
\left(  A_{1}\right)  ^{\dagger}-\left\vert \psi_{t,1}^{\prime}\right\rangle
\left\langle \psi_{t,2}^{\prime}\right\vert
\end{array}
\right]  \right\Vert _{1}\\
&  \leq2\left\Vert \mathbf{1}-A_{1}\right\Vert \left\Vert \psi_{t,1}^{\prime
}\right\Vert \left\Vert \psi_{t,2}^{\prime}\right\Vert +\left(  \left\Vert
A_{1}\right\Vert ^{2}+1\right)  \left\Vert \psi_{t,2}^{\prime}\right\Vert
^{2}\\
&  \leq6\left\Vert \psi_{t,2}^{\prime}\right\Vert \\
&  \leq C\sqrt{\left\Vert \psi_{t-1}\right\Vert ^{2}-\left\Vert \psi
_{t}\right\Vert ^{2}},
\end{align*}
where we had used (\ref{phi_2}) to show the last inequality.

Therefore,
\begin{align*}
&  \left\Vert \mathrm{tr}\,_{\mathcal{H}_{B}}\left\vert \Phi_{0,T}%
\right\rangle \left\langle \Phi_{0,T}\right\vert -\mathrm{tr}\,_{\mathcal{H}%
_{B}}\left\vert \Phi_{1,T}\right\rangle \left\langle \Phi_{1,T}\right\vert
\right\Vert _{1}\\
&  \leq\left\Vert \left\vert \varphi_{T}\right\rangle \left\langle \varphi
_{T}\right\vert -\left\vert \psi_{T}\right\rangle \left\langle \psi
_{T}\right\vert \right\Vert _{1}+1-\left\Vert \psi_{T}\right\Vert ^{2}\\
&  \leq C\sum_{t=1}^{T}\sqrt{\left\Vert \psi_{t-1}\right\Vert ^{2}-\left\Vert
\psi_{t}\right\Vert ^{2}}+1-\left\Vert \psi_{T}\right\Vert ^{2}\\
&  \leq C\sqrt{T\sum_{t=1}^{T}\left(  \left\Vert \psi_{t-1}\right\Vert
^{2}-\left\Vert \psi_{t}\right\Vert ^{2}\right)  }+1-\left\Vert \psi
_{T}\right\Vert ^{2}\\
&  =C\sqrt{T\left(  \left\Vert \psi_{0}\right\Vert ^{2}-\left\Vert \psi
_{T}\right\Vert ^{2}\right)  }+1-\left\Vert \psi_{T}\right\Vert ^{2}\\
&  =C\sqrt{\left(  1-\left\Vert \psi_{T}\right\Vert ^{2}\right)
T}+1-\left\Vert \psi_{T}\right\Vert ^{2}\\
&  \leq C\sqrt{T\,D_{T}}.+D_{T}\\
&  \leq\left(  C+1\right)  \sqrt{T\,D_{T}}%
\end{align*}
After all, to distinguish two operations, we have to have%
\begin{align*}
1-\varepsilon &  \leq\left\Vert \mathrm{tr}\,_{\mathcal{H}_{B}}\left\vert
\Phi_{0,T}\right\rangle \left\langle \Phi_{0,T}\right\vert -\mathrm{tr}%
\,_{\mathcal{H}_{B}}\left\vert \Phi_{1,T}\right\rangle \left\langle \Phi
_{1,T}\right\vert \right\Vert _{1}\\
&  \leq\left(  C+1\right)  \sqrt{T\,D_{T}},
\end{align*}
that leads to (\ref{D>1/T}).

\section{The second scenario}

The problem treated in this section is interaction-free detection of unitary
operations, where the unitary operation is chosen from the family $\left\{
U_{\theta}\right\}  _{\theta\in\Theta}$ of unitary transforms over
$\mathcal{H}_{I}$ with $k:=\left\vert \Theta\right\vert <\infty$ and
$d:=\dim\,\mathcal{H}_{I}<\infty$ . Here, the blackbox $U_{\theta}$ is given
in the form of controlled operation,%
\[
c-U_{\theta}:=\left[
\begin{array}
[c]{cc}%
\mathbf{1} & 0\\
0 & U_{\theta}%
\end{array}
\right]  ,
\]
thus making total phase meaningful. Interaction-free means that the reduced
state on $\mathcal{H}_{I}$ is unchanged throughout the process, modulo unitary
transforms independent of $\theta$. (Since this part can be canceled if
necessary.) The control space is denoted by $\mathcal{H}_{C}$, which
corresponds to which-path information. Without loss of generality, we suppose
$U_{\theta_{0}}=\mathbf{1}$. ( If this is not the case, we apply control
$U_{\theta_{0}}^{\dagger}$ right before the given control operation.)

When $U_{\theta}$'s are commutative, the standard phase estimation protocol is
sufficient to achieve the purpose: One can estimate eigenvalue arbitrary
accuracy by inflating the number of qubits (see Chapter\thinspace5 of
\cite{Nielsen-Chuang} for example).

The general case is reduced to this commutative case. Define, for each unitary
operator $U$,
\[
X_{U,\vec{\lambda}}^{k}:=\sum_{i=0}^{d-1}e^{\sqrt{-1}\lambda_{i}}\left\vert
e_{i+k}\right\rangle \left\langle e_{i}\right\vert ,\,
\]
where '+' in the subscript of $e$ is in the sense of modulo $d$, $\vec
{\lambda}=\left(  \lambda^{1},\cdots,\lambda^{d}\right)  $ and $\left\{
\left\vert e_{i}\right\rangle ;i=0,\cdots,d-1\right\}  $ are the phase of
eigenvalues and eigenvectors of $U$. Observe that the $i$-th eigenvalue of
$\left(  X_{U,\vec{\lambda}}^{k}\right)  ^{\dag}U^{\prime}X_{U,\vec{\lambda}%
}^{k}$ is $\left(  i-k\right)  $-th eigenvalue of $U^{\prime}$, if
eigenvectors of $U^{\prime}$ are also $\left\{  \left\vert e_{i}\right\rangle
;i=0,\cdots,d-1\right\}  $. Hence with $\Lambda:=\left(  \vec{\lambda}%
_{0},\vec{\lambda}_{1},\cdots\vec{\lambda}_{d-1}\right)  $ and
\[
F_{U,\Lambda}\left(  U^{\prime}\right)  :=\left(  X_{U,\vec{\lambda}_{d-1}%
}^{d-1}\right)  ^{\dag}U^{\prime}X_{U,\vec{\lambda}_{d-1}}^{d-1}\cdot
\ldots\cdot\left(  X_{U,\vec{\lambda}_{1}}^{1}\right)  ^{\dag}U^{\prime
}X_{U,\vec{\lambda}_{1}}^{1}\cdot\left(  X_{U,\vec{\lambda}_{0}}^{0}\right)
^{\dag}U^{\prime}X_{U,\vec{\lambda}_{0}}^{0},
\]
we have
\[
F_{U,\Lambda}\left(  U^{\prime}\right)  =\left(  \det U^{\prime}\right)
\mathbf{1}_{T}\mathbf{.}%
\]
In particular, $F_{U,\vec{\lambda}}\left(  \mathbf{1}\right)  =\mathbf{1}$.
Therefore, acting $F_{U,\Lambda}$ on the target space $\mathcal{H}_{I}$, \
\[
\mathbf{I}_{C}\otimes F_{U,\Lambda}\left(  c-U^{\prime}\right)  =F_{U,\Lambda
}\left(  \mathbf{1}_{T}\right)  \oplus F_{U,\Lambda}\left(  U^{\prime}\right)
=\left(  \det U^{\prime}\right)  \mathbf{1}_{C}\mathbf{\otimes1}_{T}.
\]
Thus, this transform can be implementable by acting $X_{U_{\theta}}$'s and
\ $Z_{\vec{\lambda},U_{\theta}}$'s acting on the target space $\mathcal{H}%
_{I}$.

Our basic idea is as follows. Given a black box $c-U$, we transform it to
$c-U^{\prime}$, where
\[
U^{\prime}=F_{U_{\theta_{i}},\Lambda_{i}}\circ\ldots\circ F_{U_{\theta_{2}%
},\Lambda_{2}}\circ F_{U_{\theta_{1}},\Lambda_{1}}\left(  U\right)  ,
\]
with proper choice of $\theta_{1}$, $\Lambda_{1}$, $\theta_{2}$, $\Lambda_{2}%
$, $\cdots$. \ \ $\theta_{1}$, $\Lambda_{1}$, $\theta_{2}$, $\Lambda_{2}$,
$\cdots$ are chosen so that all the members of the family
\begin{equation}
\left\{  F_{U_{\theta_{i}},\Lambda_{i}}\circ\ldots\circ F_{U_{\theta_{2}%
},\Lambda_{2}}\circ F_{U_{\theta_{1}},\Lambda_{1}}\left(  U_{\theta}\right)
\right\}  _{\theta\in\Theta}\label{new-family}%
\end{equation}
are commutative, and exactly one element, let it be $U_{\theta_{\ast}}$, is
not scalar, \textit{i.e.}, not the constant multiple of the identity. This
means that at least one eigenvalue of $U_{\theta_{\ast}}$(The latter condition
is sufficient to make sure that this family contain at least two distinct operations.)

Then we run the circuit of the phase estimation, and do the projective
measurement that judges $\theta\neq\theta_{\ast}$ without any error, and
$\theta=\theta_{\ast}$ with small error, which can be made arbitrarily small.
If the estimate is $\theta_{\ast}$, we terminate and let $\theta_{\ast}$ be
the final estimate. (In this case $\theta\neq\theta_{\ast}$ happens with small
probability, but it can be made arbitrarily small.) Otherwise, we apply the
process above to the smaller family $\Theta/\left\{  \theta_{\ast}\right\}  $
. (In this case, the true value of $\theta$ is one of $\Theta/\left\{
\theta_{\ast}\right\}  $ with certainty.) \ 

For (\ref{new-family}) to satisfy the requirements, $\theta_{1}$, $\Lambda
_{1}$, $\theta_{2}$, $\Lambda_{2}$, $\cdots$ are chosen as follows. Pick up
$U_{\theta_{1}}$ that is not a constant multiple of $\mathbf{1}$,\ and let
$U_{\theta^{\prime}}$ be the one which does not commute with $U_{\theta_{1}}$.
(If there is no such $\theta^{\prime}$, no preprocessing is necessary.) \ Then
we use the following lemma:

\begin{lemma}
There is $\Lambda$ with $F_{U,\Lambda}\left(  U^{\prime}\right)  \neq
c\mathbf{1}$ for all $U^{\prime}$ with $\left[  U,U^{\prime}\right]  \neq0$. \ 
\end{lemma}

\begin{proof}
Suppose the contrary is true, i.e., for any $\Lambda$
\[
F_{U,\Lambda}\left(  U^{\prime}\right)  =c_{\Lambda}\mathbf{1}.
\]
Taking average with respect to $\vec{\lambda}_{1},\cdots\vec{\lambda}_{d-1}$
according to the uniform dstribution, $\ $%
\[
\mathrm{E}\left[  \left(  X_{U,\vec{\lambda}_{d-1}}^{d-1}\right)  ^{\dag
}U^{\prime}X_{U,\vec{\lambda}_{d-1}}^{d-1}\cdot\ldots\cdot\left(
X_{U,\vec{\lambda}_{1}}^{1}\right)  ^{\dag}U^{\prime}X_{U,\vec{\lambda}_{1}%
}^{1}\right]
\]
is diagonal with respect to the basis $\left\{  \left\vert e_{i}\right\rangle
;i=0,\cdots,d-1\right\}  $. Therefore, by
\[
\mathrm{E}\left[  \left(  X_{U,\vec{\lambda}_{d-1}}^{d-1}\right)  ^{\dag
}U^{\prime}X_{U,\vec{\lambda}_{d-1}}^{d-1}\cdot\ldots\cdot\left(
X_{U,\vec{\lambda}_{1}}^{1}\right)  ^{\dag}U^{\prime}X_{U,\vec{\lambda}_{1}%
}^{1}\right]  \cdot\left(  X_{U,\vec{\lambda}_{0}}^{0}\right)  ^{\dag
}U^{\prime}X_{U,\vec{\lambda}_{0}}^{0}=\mathrm{E}\left[  c_{\Lambda}\right]
\mathbf{1,}%
\]
$\left(  X_{U,\vec{\lambda}_{0}}^{0}\right)  ^{\dag}U^{\prime}X_{U,\vec
{\lambda}_{0}}^{0}$ is diagonal for all $\vec{\lambda}_{0}$. This can be true
only if $U^{\prime}$ is diagonal, contradicting with the assumption $\left[
U,U^{\prime}\right]  \neq0$. Therefore, there is at least one $\Lambda$ with
$F_{U_{\theta_{1}},\Lambda}\left(  U_{\theta^{\prime}}\right)  \neq
c\mathbf{1}$.
\end{proof}

By this lemma, there is $\Lambda_{1}$ such that $F_{U_{\theta_{1}},\Lambda
_{1}}\left(  U_{\theta^{\prime}}\right)  $ is not a constant multiple of
$\mathbf{1}$. The number of elements not constant multiples of $\mathbf{1}$ in
the family $\left\{  F_{U_{\theta_{1}},\Lambda_{1}}\left(  U_{\theta}\right)
\right\}  _{\theta\in\Theta}$ is non-zero, and smaller than the number of such
elements in $\left\{  U_{\theta}\right\}  _{\theta\in\Theta}$. Repeating this
process for many times, we obtain (\ref{new-family}) that satisfies the requirements.


\begin{thebibliography}{9}                                                                                                %
\bibitem {Ambainis}A. Ambainis, "Quantum lower bounds by quantum arguments,"
STOC '00 Proceedings of the thirty-second annual ACM symposium on Theory of
computing (2000)l

\bibitem {Kwiat-etal}P.\thinspace Kwiat, H.\thinspace Weinfurter, T.\thinspace
Herzog, and A. Zeilinger, "Interactio-free Measurement," Physical Review
Letters, Vol.\thinspace74, No.\thinspace24, pp. 4763-4776 (1994)

\bibitem {Nielsen-Chuang}M.\thinspace Nielsen and M.\thinspace Chuang,
\textit{Quantum Computation and Quantum Information}, Cambridge University
Press, Cambridge (2000).
\end{thebibliography}
\end{document}